\documentclass[11pt]{article}
\usepackage{times,latexsym,graphicx,epsfig,amsmath,amsfonts,amsthm}

\newtheoremstyle{localthm}
	{5pt} 
	{5pt} 
	{\sl} 
	{} 
	{\bf} 
	{{\rm.}} 
	{.7em} 
	{} 

\theoremstyle{localthm}
\newtheorem{Theorem}{Theorem}

\newtheorem{Lemma}[Theorem]{Lemma}

\newtheoremstyle{localrem}
	{5pt} 
	{5pt} 
	{\rm} 
	{} 
	{\bf} 
	{{\rm.}} 
	{.7em} 
	{} 

\theoremstyle{localrem}

\newcommand{\ruck}[1]{\strut\hspace{#1cm}}
\def\aSt{a_\alpha^{\rm St}}
\def\bSt{b_\alpha^{\rm St}}
\def\logit{\mathop{\rm logit}\nolimits}

\def\Bin{\mathrm{Bin}}
\def\Poiss{\mathrm{Poiss}}

\def\Var{\mathrm{Var}}

\def\Ex{\mathrm{I\!E}}
\def\Pr{\mathrm{I\!P}}

\def\R{\mathbb{R}}

\def\LL{\mathcal{L}}
\def\XX{\mathcal{X}}

\begin{document}

\addtolength{\baselineskip}{+.5\baselineskip}

\title{Exact Confidence Bounds in Discrete Models\\ --
	Algorithmic Aspects of Sterne's Method}

\author{Lutz D\"umbgen\\University of Bern}

\date{December 2021}

\maketitle

\vfill

\paragraph{Abstract.}
In this manuscript we review two methods to construct exact confidence bounds for an unknown real parameter in a general class of discrete statistical models. These models include the binomial family, the Poisson family as well as distributions connected to odds ratios in two-by-two tables. In particular, we discuss Sterne's (1954) method in our general framework and present an explicit algorithm for the computation of the resulting confidence bounds. The methods are illustrated with various examples.

\vfill

\paragraph{AMS 2000 subject classification:}
62F25.

\paragraph{Key words:}
Binomial distribution, contingency table, discrete exponential family, log-concavity, Poisson distribution, nested binary search, odds ratio, two-by-two table.

\paragraph{Comment:}
This is an updated version of a mansucript from 2004.

\newpage

\section{Introduction}

Optimal statistical testing is well-understood in exponential families of distributions; see for instance Lehmann~(1986). This theory entails optimal one-sided confidence bounds for an unknown real parameter. However, at least in case of discrete distributions such as the binomial or Poisson there exist various proposals for computation of confidence intervals.

For the particular case of a binomial distribution $\Bin(n,p)$ with given sample size and an unknown parameter $p \in [0,1]$ there is an excellent review of approximate and exact coufidence bounds by Brown et al.\ (2001). Here we just mention the approximate bounds of Wilson (1927), the exact bounds of Clopper and Pearson (1934), Casella's (1986) refinement of the latter method and Sterne's (1954) approach; see also Crow (1956), Clunies-Ross (1958) and Blyth and Still (1983).

In Section~\ref{Setting} of the present paper we describe the general setting, namely, discrete single-parameter exponential families, generated by log-concave probability weights. In Section~\ref{Methods} we review one-sided confidence bounds and the resulting confidence intervals of Clopper and Pearson (1934). Moreover we describe Sterne's (1954) proposal. Section~\ref{Sterne} analyzes Sterne's method in more detail and presents a general algorithm for its implementation. The latter is an alternative to an algorithm of Baptista and Pike (1977), who considered the special case of odds ratios. In Section~\ref{Examples} we illustrate and compare the three versions of confidence intervals with various models and data examples. Longer proofs are deferred to Section~\ref{Proofs}.

\section{The general setting}
\label{Setting}

Throughout this paper we consider the following exponential family of discrete probability distributions. Let $\XX$ be an interval of integers, and for $x \in \XX$ let $w_x > 0$. We assume that
\[
	\sum_{x \in \XX} w_x e^{\theta x} < \infty
	\quad\text{for all} \ \theta \in \R .
\]
For any $\theta \in \R$ we define probability weights
\[
	f_\theta(x) \ := \ \frac{w_x e^{\theta x}}{\sum_{y \in \XX} w_y e^{\theta y}} ,
	\quad	x \in \XX .
\]
The corresponding distribution and cumulative distribution function are denoted by $P_\theta$ and $F_\theta$, respectively.

In case of $\inf(\XX) \in \XX$, we also include the parameter $\theta = -\infty$ and define $f_{-\infty}(x) := 1_{[x = \inf(\XX)]}$. Analogously, if $\sup(\XX) \in \XX$, we include $\theta = \infty$ and set $f_{\infty}(x) := 1_{[x = \sup(\XX)]}$. This leads to a parameter space $\Theta$ such that $\R \subset \Theta \subset [-\infty,\infty]$. In connection with Sterne's method we assume in addition that
\begin{equation}
\label{eq:log-concavity}
	\frac{w_{x+1}}{w_x} \ \
	\text{is strictly decreasing in} \
	x \in \XX ,
\end{equation}
where we set $w_z := 0$ for $z \in \mathbb{Z} \setminus \XX$. The reason for this assumption will become clear later. It entails that for any $\theta \in \mathbb{R}$,
\[
	\log f_\theta(x) \ = \ \mathrm{const}(\theta) + \log w_x + \theta x
\]
is strictly concave in $x \in \XX$.

Our goal is to construct confidence regions for $\theta$. Before we describe the methods, let us mention three examples for the present setting.

\paragraph{Binomial distributions.}
The distributions $\Bin(n,p)$, $p \in [0,1]$, fit into our framework with
\[
	\XX \ = \ \{0,1,\ldots,n\} ,	\quad
	w_x \ = \ \binom{n}{x} ,	\quad
	\theta \ = \ \logit(p) := \log \Bigl( \frac{p}{1 - p} \Bigr) \in [-\infty,\infty] .
\]

\paragraph{Poisson distributions.}
Another example are the distributions $\Poiss(\lambda)$, $\lambda \ge 0$. Here
\[
	\XX \ = \ \{0,1,2,\ldots\} ,	\quad
	w_x \ = \ \frac{1}{x!} ,	\quad
	\theta \ = \ \log(\lambda) \in [-\infty,\infty) .
\]

\paragraph{Odds ratios.}
Let $Y_1 \sim \Bin(n_1,p_1)$ and $Y_2 \sim \Bin(n_2,p_2)$ be stochastically independent, where $p_1, p_2 \in (0,1)$. Then the conditional distribution
\[
	\LL(Y_1 \,|\, Y_1 + Y_2 = s)
\]
is of the general type above, where
\[
	\XX \ = \ \bigl\{ \max(0,s - n_2), \ldots, \min(n_1,s) \bigr\} ,
	\quad
	w_x \ = \ \frac{1}{x! (n_1-x)! (s - x)! (n_2 - s + x)!} ,
\]
and $\theta$ is the log-odds ratio
\[
	\theta \
	= \ \log \Bigl( \frac{p_1(1 - p_2)}{(1 - p_1) p_2} \Bigr) \
	= \ \logit(p_1) - \logit(p_2) \in \R .
\]
This type of distribution arises in various applications involving two dichotomous variables and two-by-two tables.

\section{Three approaches to confidence sets}
\label{Methods}

In what follows let $X$ be a random variable with distribution $P_\theta$ and distribution function $F_\theta$, where $\theta$ is an unknown parameter in $\Theta$. Our goal is to construct a confidence set for $\theta$ with confidence level at least $1 - \alpha$ for a given $\alpha \in (0,1)$. That means, we are looking for a mapping $\XX \ni x \mapsto C_\alpha(x) \subset \Theta$ such that
\[
	\Pr_\eta(C_\alpha(X) \ni \eta) \ \ge \ 1 - \alpha
	\quad\text{for any} \ \eta \in \Theta .
\]
Here and in the sequel, probabilities and expectations in case of $\theta = \eta$ are denoted by $\Pr_\eta$ and $\Ex_\eta$, respectively.

\paragraph{One-sided confidence bounds.}
For a detailed and complete discussion of one-sided tests and confidence bounds we refer the reader to Lehmann~(1986). Here is just a description of the resulting confidence bounds in the present setting.

It is well-known that for any $\alpha \in (0,1)$ and $\eta \in \Theta$,
\[
	\Pr_\eta(F_\eta(X) > \alpha) \ \ge \ 1 - \alpha .
\]
Moreover, for fixed $x \in \XX \setminus \{\sup(\XX)\}$, the function $\Theta \ni \eta \mapsto F_\eta(x)$ is continuous and strictly decreasing with limits $F_{-\infty}(x) = 1$ and $F_{\infty}(x) = 0$. Consequently, a $(1 - \alpha)$-confidence region for $\theta$ is given by
\[
	C_\alpha(x) \
	:= \ \bigl\{ \eta \in \Theta : F_\eta(x) \ge \alpha \bigr\} \
	 = \ \Theta \cap \bigl[ -\infty, b_\alpha(x) \bigr] ,
\]
where $b_\alpha(x) := \infty$ in case of $x = \sup(\XX)$, while $b_\alpha(x)$ is the unique solution $\eta \in \Theta$ of the equation $F_\eta(x) = \alpha$ in case of $x < \sup(\XX)$.

Analogously, since
\[
	\Pr_\eta( F_\eta(X-1) < 1 - \alpha) \ \ge \ 1 - \alpha
\]
for any $\alpha \in (0,1)$ and $\eta \in \Theta$, an alternative $(1 - \alpha)$-confidence region for $\theta$ is given by
\[
	C_\alpha(x) \
	:= \ \bigl\{ \eta \in \Theta : F_\eta(x-1) \le 1 - \alpha \bigr\} \
	 = \ \Theta \cap \bigl[ a_\alpha(x), \infty \bigr] ,
\]
where $a_\alpha(x) := -\infty$ in case of $x = \inf(\XX)$, while $a_\alpha(x)$ is the unique solution $\eta \in \Theta$ of the equation $F_\eta(x-1) = 1-\alpha$ in case of $x > \inf(\XX)$.

An alternative representation of these confidence regions is in terms of p-values: The left-sided p-value for the null hypothesis that $\theta = \eta$ is given by $\pi^{\rm left}(X,\eta) := F_\eta(X)$, and
\[
	\Theta \cap \bigl[ -\infty, b_\alpha(X) \bigr] \
	= \ \bigl\{ \eta \in \Theta : \pi^{\rm left}(X,\eta) \ge \alpha \bigr\} .
\]
Analogously, the right-sided p-value for the null hypothesis that $\theta = \eta$ equals $\pi^{\rm right}(X,\eta) := 1 - F_\eta(X-1)$, and
\[
	\Theta \cap \bigl[ a_\alpha(X), \infty \bigr] \
	= \ \bigl\{ \eta \in \Theta : \pi^{\rm right}(X,\eta) \ge \alpha \bigr\} .
\]

\paragraph{Clopper and Pearson's confidence interval.}
Clopper and Pearson (1934) proposed to combine the lower and upper confidence bound with $\alpha/2$ in place of $\alpha$. This yields the $(1 - \alpha)$-confidence interval
\[
	C_\alpha(x) := \ \Theta \cap \bigl[ a_{\alpha/2}(x), b_{\alpha/2}(x) \bigr]
\]
for $\theta$. An alternative representation of this confidence interval is in terms of the two-sided p-value $\pi^{\rm two}(X,\eta) := 2 \min \bigl\{ \pi^{\rm left}(X,\eta), \pi^{\rm right}(X,\eta) \bigr\}$ is given by
\[
	\Theta \cap \bigl[ a_{\alpha/2}(X), b_{\alpha/2}(X) \bigr]
	\ = \ \bigl\{ \eta \in \Theta : \pi^{\rm two}(X,\eta) \ge \alpha \bigr\} .
\]
This method is easy to implement and offered by various software packages. However, these intervals tend to be too conservative. Various refinements have been proposed in special cases, especially for the binomial model; see for instance Casella (1986).

\paragraph{Sterne's confidence sets.}
An alternative and intuitively appealing approach is to choose for each hypothetical parameter $\eta \in \Theta$ an acceptance region $A_\eta \subset \XX$ with minimal cardinality such that
\[
	\sum_{x \in A_\eta} f_\eta(x) \ \ge \ 1 - \alpha .
\]
One can easily verify that $A_\eta$ may be constructed by successively adding points $x \in \XX$ with maximal probability weight $f_\eta(x)$. The resulting $(1 - \alpha)$-confidence set equals $C_\alpha(x) := \bigl\{ \eta \in \Theta : x \in A_\eta \bigr\}$.

This construction is essentially equivalent to the following procedure: For any hypothetical parameter $\eta \in \Theta$ and potential observation $x \in \XX$ define the p-value
\[
	\pi(x,\eta) \ := \ \sum_{y \in \XX} 1_{[f_\eta(y) \le f_\eta(x)]} f_\eta(y)
\]
for the null hypothesis that $\theta = \eta$. This yields the $(1 - \alpha)$-confidence set
\[
	C_\alpha^{\rm St}(x) \ := \ \{\eta \in \Theta : \pi(x,\eta) > \alpha\} .
\]

A problem with this approach is that the p-value function $\pi(x,\cdot)$ need not be unimodal. For instance, in case of Poisson distributions it never is; see Section~\ref{Examples}. Nevertheless one can compute the boundaries
\[
	\aSt(x) \ := \ \inf(C_\alpha^{\rm St}(x))
	\quad\text{and}\quad
	\bSt(x) \ := \ \sup(C_\alpha^{\rm St}(x))
\]
as explained in Section~\ref{Sterne} and obtains the $(1 - \alpha)$-confidence interval $\bigl[ \aSt(x), \bSt(x) \bigr]$ for $\theta$.

\section{Computation of Sterne's bounds}
\label{Sterne}

\subsection{Analytical properties of $\pi(x,\cdot)$}

Note that strict concavity of $\XX \ni x \mapsto \log w_x$ implies that the special parameters
\[
	\theta_{k,x} \ := \ \frac{\log w_x - \log w_k}{k - x}
\]
are strictly increasing in $k \in \XX \setminus \{x\}$. In addition we introduce $\theta_{\inf(\XX)-1,x} := -\infty$ and $\theta_{\sup(\XX)+1,x} := \infty$. These parameters $\theta_{k,x}$ play an important role in the following theorem, which summarizes crucial properties of the p-value function $\pi(x,\cdot)$. 

\begin{Theorem}
\label{Analysis}

\textbf{(a)} \ The function $\pi(x,\cdot)$ has discontinuities at the points $\theta_{k,x}$, $k \in \XX \setminus \{x\}$. Precisely,
\[
	\pi(x,\eta) \ = \ \begin{cases}
		1 & \text{if} \ \theta_{x-1,x} \le \eta \le \theta_{x+1,x} , \\
		1 - F_\eta(k-1) + F_\eta(x)
		  & \text{if} \ k > x+1 \ \text{and} \ \theta_{k-1,x} < \eta \le \theta_{k,x} , \\
		F_\eta(x)
		  & \text{if} \ \sup(\XX) < \infty \ \text{and} \ \eta > \theta_{\sup(\XX),x} , \\
		1 - F_\eta(x-1) + F_\eta(k)
		  & \text{if} \ k < x-1 \ \text{and} \ \theta_{k,x} \le \eta < \theta_{k+1,x} , \\
		1 - F_\eta(x-1)
		  & \text{if} \ \inf(\XX) > -\infty \ \text{and} \ \eta < \theta_{\inf(\XX),x} .
	\end{cases}
\]

\noindent
\textbf{(b)} \ $\pi(x,\theta_{k,x})$ is strictly decreasing in $k > x$ and strictly increasing in $k < x$.

\noindent
\textbf{(c)} \ For $k \in \XX \setminus \{x-1,x,x+1\}$ there exists a strictly concave function $\psi_{k,x}^{} : \mathbb{R} \to \mathbb{R}$ such that
\[
	\pi(x,\eta) \ = \ 1 - \exp(\psi_{k,x}^{}(\eta))
	\quad\text{for}\quad
	\eta \ \in \ \begin{cases}
		[\theta_{k,x}, \theta_{k+1,x}) & \text{if} \ k < x-1 , \\\\
		(\theta_{k-1,x}, \theta_{k,x}] & \text{if} \ k > x+1 .
	\end{cases}
\]
\end{Theorem}

Note that part~(a) of this theorem implies that at each discontinuity point $\theta_{k,x}$ the p-value function $\pi(x,\cdot)$ behaves as follows:
\[
	\pi(x,\theta_{k,x}) \ > \ \begin{cases}
		\pi(x,\theta_{k,x} \, +) & \text{if} \ k > x , \\
		\pi(x,\theta_{k,x} \, -) & \text{if} \ k < x ,
	\end{cases}
\]
where $\pi(x,\theta_{k,x}\,\pm)$ denotes the right-sided and left-sided limit of $\pi(x,\cdot)$ at $\theta_{k,x}$, respectively.

\subsection{Computation of $\bSt(x)$}

Since computing the lower confidence bound $\aSt(x)$ is analogous to computing the upper bound $\bSt(x)$, we describe the general procedure only for the latter one.

At first we assume that the support set $\XX$ is bounded. Specific comments about the case of unbounded $\XX$ are made later.

\paragraph{First stage.}
Note first that $\bSt(x) = \infty$ in case of $x = \max(\XX)$, because then $\pi(x,\infty) = 1$.

Thus we restrict our attention to the case that $x < \max(\XX)$. It follows from Theorem~\ref{Analysis}~(a-b) that there exists an integer $k_\alpha(x) \in \XX$ with $k_\alpha(x) \ge x+1$ such that
\[
	\pi(x,\theta_{k,x}^{}) \ \begin{cases}
		\ge \ \alpha & \text{if} \ k \in \{x+1, \ldots, k_\alpha(x)\} , \\
		<   \ \alpha & \text{if} \ k > k_\alpha(x), k \in \XX .
	\end{cases}
\]
This integer may be obtained by a naive linear search or a suitable binary search. Table~\ref{Stage One} contains pseudocode for the determination of $k_\alpha(x)$.

\begin{table}
\centering
\begin{tabular}{|l|} \hline
\ruck{0}	$k' \leftarrow \max(\XX)$\\
\ruck{0}	$\pi' \leftarrow \pi(x,\theta_{k',x})$\\
\ruck{0}	if $\pi' \ge \alpha$ then\\
\ruck{1}		$k \leftarrow k'$\\
\ruck{0}	else\\
\ruck{1}		$k \leftarrow x+1$\\
\ruck{1}		$\pi \leftarrow 1$\\
\ruck{1}		while $k' > k+1$ do\\
\ruck{2}			$k'' \leftarrow \lfloor (k + k')/2 \rfloor$\\
\ruck{2}			$\pi'' \leftarrow \pi(x,\theta_{k'',x})$\\
\ruck{2}			if $\pi'' \ge \alpha$ then\\
\ruck{3}				$k \leftarrow  k''$\\
\ruck{2}			else\\
\ruck{3}				$k' \leftarrow  k''$\\
\ruck{2}			end if\\
\ruck{1}		end while\\
\ruck{0}	end if\\\hline
\end{tabular}
\caption{Computation of $k = k_\alpha(x)$.}
\label{Stage One}
\end{table}

\paragraph{Second stage.}
If $k_\alpha(x) = \max(\XX)$, then
\[
	\pi(x,\eta) \ = \ F_\eta(x)
	\quad\text{for} \ \theta_{\max(\XX),x} < \eta < \infty .
\]
Thus
\[
	\bSt(x) \ = \ \begin{cases}
		\theta_{\max(\XX),x} & \text{if} \ F_{\theta_{\max(\XX),x}}(x) \le \alpha , \\
		b_\alpha(x)   & \text{else} .
	\end{cases}
\]

Now we consider the case that $k = k_\alpha(x) < \max(\XX)$. According to Theorem~\ref{Analysis}~(c), $\pi(x,\eta)$ may be written as $1 - \exp(\psi(\eta))$ for some concave function $\psi = \psi_{k+1,x}$ on $(\theta_{k,x}, \theta_{k+1,x}]$, where $\pi(x,\theta_{k+1, x}) < \alpha$. Thus, if the right-sided limit
\[
	\pi(x, \theta_{k,x} \, +)
	\ = \ 1 - \exp(\psi(\theta_{k,x}))
	\ = \ 1 - F_{\theta_{k,x}}(k) + F_{\theta_{k,x}}(x)
\]
is not greater than $\alpha$, then $\pi(x, \eta) < \alpha$ for all $\eta > \theta_{k,x}$, whence $\bSt(x) = \theta_{k,x}$. Otherwise $\bSt(x)$ is the unique number $b \in (\theta_{k,x}, \theta_{k+1,x})$ such that $\pi(x, b) = 1 - F_a(k) + F_a(x) = \alpha$. This number can also be determined numerically by a suitable version of binary search.

Table~\ref{Stage Two} contains corresponding pseudocode. The algorithm involves a precision number $\delta > 0$, and it terminates with a number $b \in [\bSt(x), \bSt(x) + \delta]$ such that $\pi(x,b) \in [\alpha-\delta,\alpha]$.

\begin{table}
\centering
\begin{tabular}{|l|} \hline
\ruck{0}	$b'   \leftarrow \theta_{k,x}$\\
\ruck{0}	$\pi' \leftarrow 1 - F_{b'}(k) + F_{b'}(x)$\\
\ruck{0}	if $\pi' \le \alpha$ then\\
\ruck{1}		$b \leftarrow b'$\\
\ruck{0}	else\\
\ruck{1}		$b   \leftarrow \theta_{k+1,x}$\\
\ruck{1}		$\pi \leftarrow 1 - F_b(k) + F_b(x)$\\
\ruck{1}		while $b - b' > \delta$ or $\pi' - \pi > \delta$ do\\
\ruck{2}			if $b < \infty$ then\\
\ruck{3}				$b'' \leftarrow (b' + b)/2$\\
\ruck{2}			else\\
\ruck{3}				$b'' \leftarrow b' + 1$\\
\ruck{2}			end if\\
\ruck{2}			$\pi'' \leftarrow 1 - F_{b''}(k) + F_{b''}(x)$\\
\ruck{2}			if $\pi'' > \alpha$ then\\
\ruck{3}				$b'   \leftarrow b''$\\
\ruck{3}				$\pi' \leftarrow \pi''$\\
\ruck{2}			else\\
\ruck{3}				$b   \leftarrow b''$\\
\ruck{3}				$\pi \leftarrow \pi''$\\
\ruck{2}			end if\\
\ruck{1}		end while\\
\ruck{0}	end if\\\hline
\end{tabular}
\caption{Approximation $b$ of $\bSt(x)$, where $k = k_\alpha(x)$.}
\label{Stage Two}
\end{table}

\paragraph{The case of unbounded $\XX$.}
If the support set $\XX$ is unbounded, we assume in addition to the previously mentioned conditions that for any fixed $x \in \XX$,
\begin{equation}
	\lim_{|k| \to \infty, \, k \in \XX} \, \pi(x, \theta_{k,x}) \ = \ 0 .
	\label{Analysis2}
\end{equation}
Then the integer $k_\alpha(x)$ is still well-defined and may be obtained by a suitable search strategy. The second stage of the algorithm is exactly the same as before.

In the Poisson example, our only example with unbounded $\XX$, Condition~\eqref{Analysis2} is indeed satisfied. It follows from Stirling's formula that for $k > x$,
\[
	\theta_{k,x}^{} \ = \ \frac{\log(k!) - \log(x!)}{k - x}
	\ = \ \log k - 1 + o(1)
	\quad\text{as} \ k \to \infty .
\]
Therefore, as $k \to \infty$, the corresponding Poisson parameter equals $\exp(\theta_{k,x}) = k/e + o(k)$, and it follows from Tshebyshev's inequality that
\begin{align*}
	\pi(x,\theta_{k,x}) \
	&= \ 1 - \Poiss(\exp(\theta_{k,x})) \bigl( \{x+1,\ldots,k-1\} \bigr) \\
	&\le \ \Poiss(\exp(\theta_{k,x})) \bigl( \bigl\{ j \in \XX
		: |j - \exp(\theta_{k,x})| \ge k/e + o(k) \bigr\} \bigr) \\
	&\le \ \frac{\Var(\Poiss(\exp(\theta_{k,x})))}{(k/e + o(k))^2} \\
	&= \ O(k^{-1}) .
\end{align*}

\section{Examples}
\label{Examples}

\paragraph{Binomial parameters.}
The interval $[\theta_{x-1,x}, \theta_{x+1,x}]$ on which $\pi(x,\cdot) \equiv 1$ equals
\[
	\Bigl[ \logit \Bigl( \frac{x}{n+1} \Bigr),
		\logit \Bigl( \frac{x+1}{n+1} \Bigr) \Bigr] ,
\]
and thus it contains the maximum likelihood estimator $\hat{\theta} = \logit(x/n)$ of $\theta = \logit(p)$.

Figure~\ref{Bin1} depicts the p-value function $p \mapsto \pi(5,\logit(p))$ for the binomial model with $n = 20$. The vertical lines show the special parameters $\logit^{-1}(\theta_{k,5})$ for $k \in \{0,1,\ldots,20\} \setminus\{5\}$. The horizonal line at height $\alpha = 0.05$ yields the 95\%-confidence bounds $\logit^{-1}(a_{0.05}^{\rm St}(5)) = 0.104$ and $\logit^{-1}(b_{0.05}^{\rm St}(5)) = 0.475$ for the underlying parameter $p = \logit^{-1}(\theta)$.

Figure~\ref{Bin2} shows for $n = 20$ and $x = 0,1,\ldots,10$ the resulting 95\%-confidence intervals for the parameter $p$ via Sterne's method (blue left bar) and Clopper-Pearson's approach (red right bar). In all cases Sterne's intervals are shorter than the competitors, but this is not a general rule.

\begin{figure}
\centering
\includegraphics[width=0.7\textwidth]{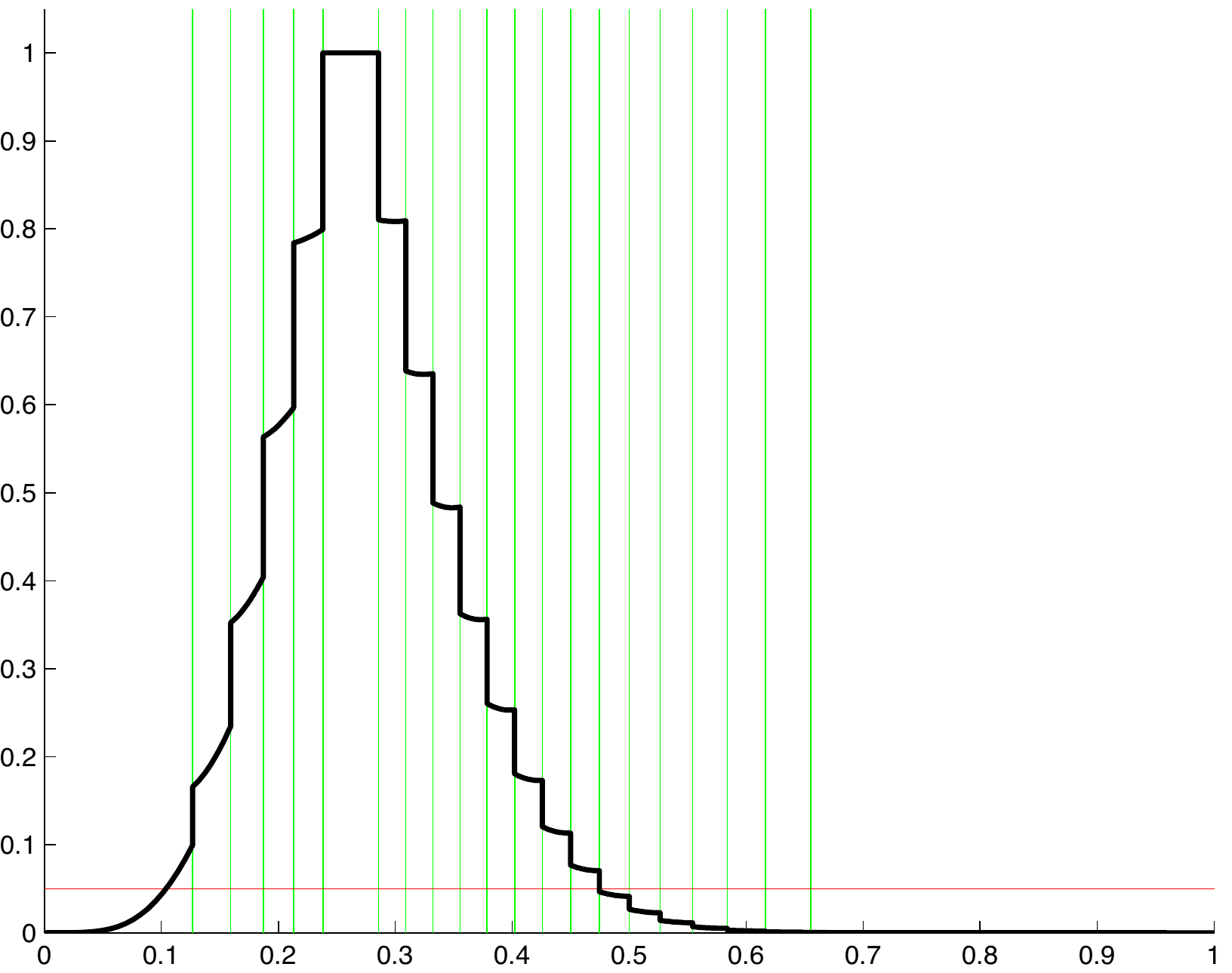}
\caption{The p-value function $p \mapsto \pi(5,\logit(p))$ for the binomial model with $n = 20$.}
\label{Bin1}
\end{figure}

\begin{figure}
\centering
\includegraphics[width=0.7\textwidth]{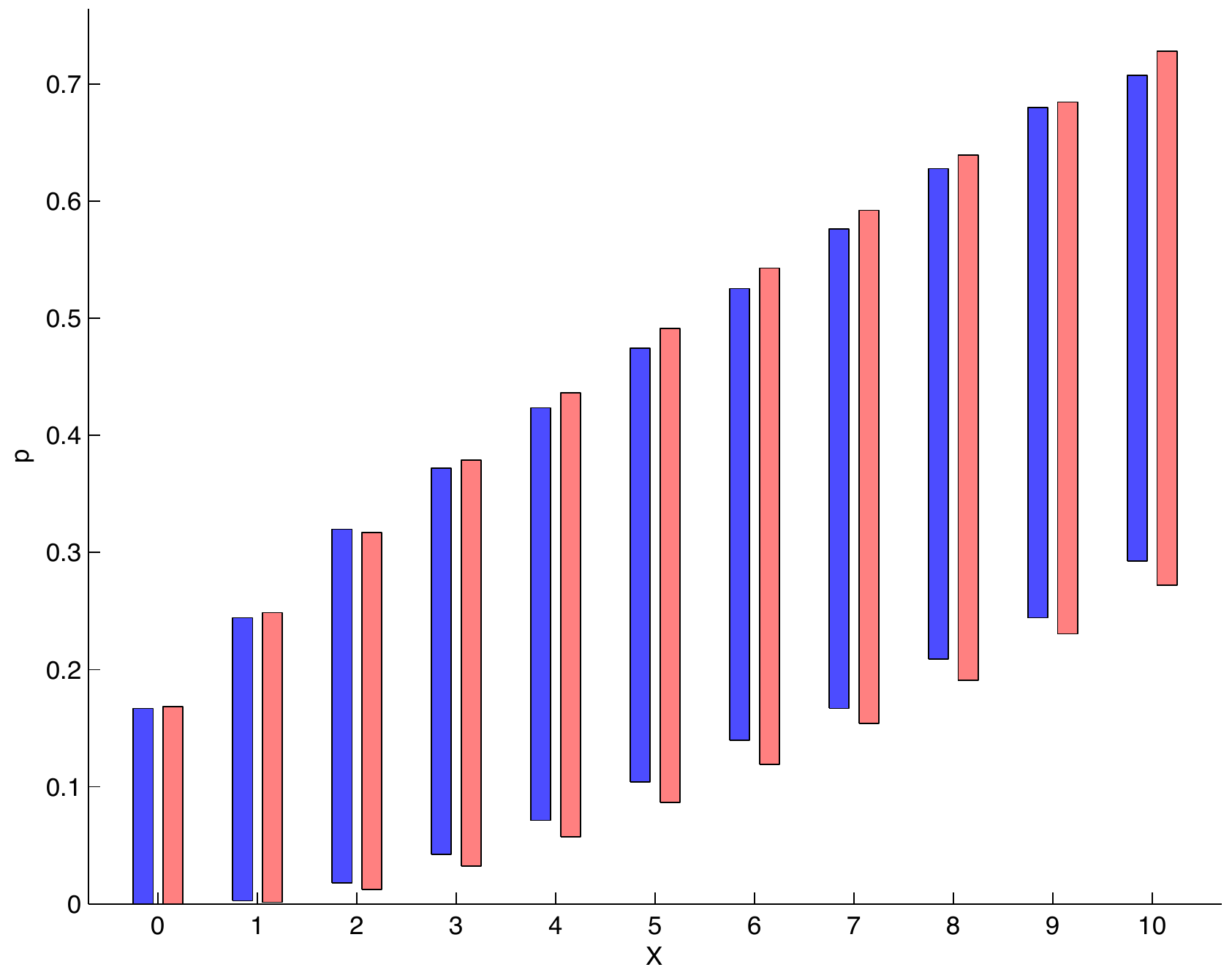}
\caption{95\%-confidence intervals for a binomial parameter $p$, when $n = 20$.}
\label{Bin2}
\end{figure}

\paragraph{Poisson parameters.}
The interval $[\theta_{x-1,x}, \theta_{x+1,x}]$ on which $\pi(x,\cdot) \equiv 1$ equals
\[
	\bigl[ \log(x), \log(x+1) \bigr] ,
\]
so its left endpoint is the maximum likelihood estimator $\hat{\theta} = \log(x)$ of $\theta = \log(\lambda)$.

The Poisson family has another special property. Namely, the p-value function $\pi(x,\cdot)$ is strictly increasing on
\[
	\begin{cases}
		[\theta_{k,x}, \theta_{k+1,x}) & \text{for} \ 0 \le k < x-1 , \\
		(\theta_{k-1,x}, \theta_{k,x}) & \text{for} \ k > x+1 .
	\end{cases}
\]
For if $0 \le k < x-1$, then for $\exp(\theta_{k,x}) \le \lambda < \exp(\theta_{k+1,x})$,
\begin{align*}
	\frac{d}{d\lambda} \, \pi(x, \log\lambda) \
	&= \ \frac{d}{d\lambda}
		\Bigl( 1 - e^{-\lambda} \sum_{j=k+1}^{x-1} \frac{\lambda^j}{j!} \Bigr) \\
	&= \ e^{-\lambda}
		\Bigl( \frac{\lambda^{x-1}}{(x-1)!} - \frac{\lambda^k}{k!} \Bigr) \\
	&= \ e^{-\lambda} \frac{\lambda^k}{k!}
		\Bigl( \frac{\lambda^{x-k-1} k!}{(x-1)!} - 1 \Bigr) ,
\end{align*}
and
\[
	\frac{\lambda^{x-k-1} k!}{(x-1)!}
	\ \ge \ \frac{\exp(\theta_{k,x})^{x-k-1} k!}{(x-1)!}
	\ = \ \exp \Bigl( \log x - \frac{1}{x - k} \sum_{j=k+1}^x \log j \Bigr)
	\ > \ 1 .
\]
Likewise, for $k > x+1$ and $\exp(\theta_{k-1,x}) < \lambda \le \exp(\theta_{k,x})$,
\[
	\frac{d}{d\lambda} \, \pi(x, \log\lambda)
	\ = \ e^{-\lambda}
		\Bigl( \frac{\lambda^{k-1}}{(k-1)!} - \frac{\lambda^x}{x!} \Bigr)
	\ = \ e^{-\lambda} \frac{\lambda^x}{x!}
		\Bigl( \frac{\lambda^{k-1-x} x!}{(k-1)!} - 1 \Bigr) ,
\]
and
\[
	\frac{\lambda^{k-1-x} x!}{(k-1)!}
	\ > \ \frac{\exp(\theta_{k-1,x})^{k-1-x} x!}{(k-1)!}
	\ = \ 1 .
\]

A consequence of this isotonicity property is that the second stage of our algorithm is superfluous when computing the upper bound $\bSt(x)$.

The interval $[\theta_{x-1,x}, \theta_{x+1,x}]$ on which $\pi(x,\cdot) \equiv 1$ equals $\bigl[ \log(x), \log(x+1) \bigr]$, so its left endpoint is the maximum likelihood estimator $\hat{\theta} = \log(x)$ of $\theta = \log(\lambda)$.

Figure~\ref{Poisson1} depicts a part of the p-value function $\lambda \mapsto \pi(3,\log\lambda)$. The vertical lines show the special parameters $\exp(\theta_{k,3})$ for $k = 0,1,2$ and $k = 4,5,\ldots$.

Table~\ref{Poisson2} contains for $x = 0,1,\ldots,15$ the resulting 95\%-confidence intervals for $\lambda = \exp(\theta)$ via Sterne's method and, in brackets, Clopper-Pearson's approach. The latter intervals are longer, except for $x = 0$ and $x = 1$.

\begin{figure}
\centering
\includegraphics[width=0.7\textwidth]{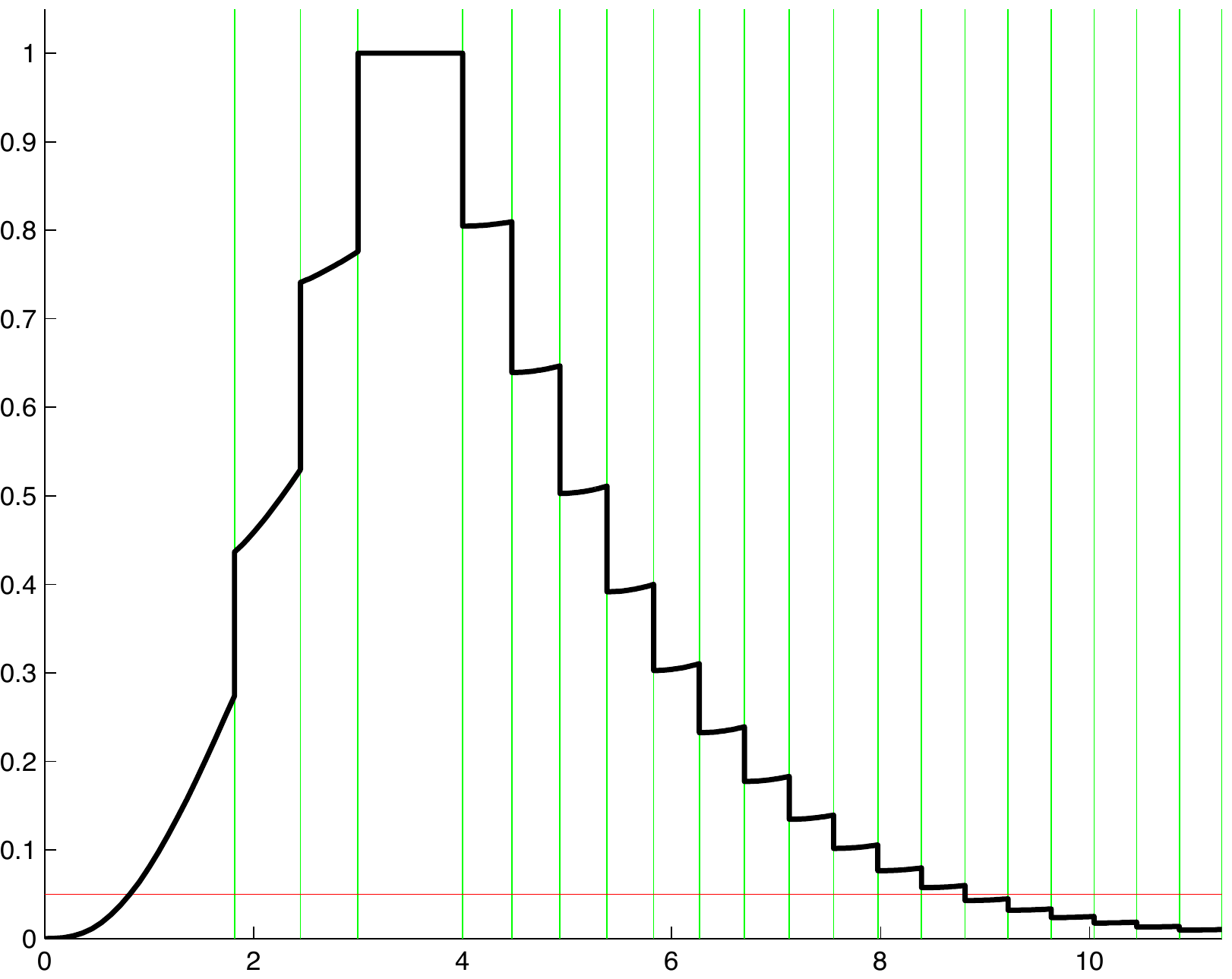}
\caption{The p-value function $\lambda \mapsto \pi(3,\log\lambda)$ for the Poisson model.}
\label{Poisson1}
\end{figure}

\begin{table}
\[
	\begin{array}{|c||cc|c|} \hline
	x & \text{lower b.} & \text{upper b.} & \text{length} \\\hline
	0 &  0.0000  &   3.7644  &   3.7644  \\
	  & (0.0000) &  (3.6889) &  (3.6889) \\\hline
	1 &  0.0512  &   5.7560  &   5.7048  \\
	  & (0.0253) &  (5.5717) &  (5.5464) \\\hline
	2 &  0.3553  &   7.2950  &   6.9397  \\
	  & (0.2422) &  (7.2247) &  (6.9825) \\\hline
	3 &  0.8176  &   8.8077  &   7.9901  \\
	  & (0.6186) &  (8.7673) &  (8.1487) \\\hline
	4 &  1.3663  &  10.3073  &   8.9410  \\
	  & (1.0898) & (10.2416) &  (9.1518) \\\hline
	5 &  1.9701  &  11.7992  &   9.8291  \\
	  & (1.6234) & (11.6684) & (10.0450) \\\hline
	6 &  2.6130  &  13.2862  &  10.6732  \\
	  & (2.2018) & (13.0595) & (10.8577) \\\hline
	7 &  3.2853  &  14.3403  &  11.0550  \\
	  & (2.8143) & (14.4227) & (11.6084) \\\hline
	\end{array}
	\qquad
	\begin{array}{|c||cc|c|} \hline
	x  & \text{lower b.} & \text{upper b.} & \text{length} \\\hline
	 8 &   3.7643  &  15.8198  &  12.0555  \\
	   &  (3.4538) & (15.7632) & (12.3094) \\\hline
	 9 &   4.4601  &  17.2979  &  12.8378  \\
	   &  (4.1153) & (17.0849) & (12.9696) \\\hline
	10 &   5.3233  &  18.3386  &  13.0153  \\
	   &  (4.7953) & (18.3904) & (13.5951) \\\hline
	11 &   5.7559  &  19.8138  &  14.0579  \\
	   &  (5.4911) & (19.6821) & (14.1910) \\\hline
	12 &   6.6857  &  20.8485  &  14.1628  \\
	   &  (6.2005) & (20.9616) & (14.7611) \\\hline
	13 &   7.2949  &  22.3219  &  15.0270  \\
	   &  (6.9219) & (22.2304) & (15.3085) \\\hline
	14 &   8.1020  &  23.7952  &  15.6932  \\
	   &  (7.6539) & (23.4897) & (15.8358) \\\hline
	15 &   8.8076  &  24.8249  &  16.0173  \\
	   &  (8.3953) & (24.7403) & (16.3450) \\\hline
	\end{array}
\]
\caption{95\%-confidence intervals for a Poisson parameter $\lambda$.}
\label{Poisson2}
\end{table}

\paragraph{Odds ratios.}
The interval $[\theta_{x-1,x}, \theta_{x+1,x}]$ on which $\pi(x,\cdot) \equiv 1$ equals
\[
	\Bigl[
		\logit \Bigl( \frac{x}{n_1+1} \Bigr)
			- \logit \Bigl( \frac{s-x+1}{n_2+1} \Bigr) ,
		\logit \Bigl( \frac{x+1}{n_1+1} \Bigr)
			- \logit \Bigl( \frac{s-x}{n_2+1} \Bigr) \Bigr] .
\]
This interval contains $\logit(\hat{\rho})$ with the point estimator
\[
	\hat{\rho} \ := \ \frac{(x+1/2)(n_1 - s + x + 1/2)}{(n_1 - x + 1/2)(s-x+1/2)}
\]
of the odds ratio
\[
	\rho \ := \ \frac{p_1/(1 - p_1)}{p_2/(1 - p_2)} .
\]

We illustrate the presented methods with a case-control study about cervical cancer by Graham and Shotz (1979). Here $n_1 = 49$ mothers aged 50-59 with diagnosed cervical cancer (cases) had been compared with $n_2 = 317$ mothers from the same age group without cancer (controls). The question was whether age at first pregnancy is associated with cervical cancer. In the case group the age of $Y_1 = 42$ mothers at first pregnancy was 25 or less, whereas in the control group the number of young mothers was $Y_2 = 203$. Viewing the observations $Y_i$ as independent random variables with distributions $\Bin(n_i,p_i)$, a point estimator for the odds ratio $\rho$ is given by $\hat\rho \approx 3.1884$. With Sterne's method we obtain the 95\%-confidence interval $[1.4427, 8.0213]$ for $\rho = \exp(\theta)$. The Clopper--Pearson approach yields the larger interval $[1.4332, 9.1586]$. Figure~\ref{OddsRatio} depicts part of the corresponding p-value function $\rho \mapsto \pi({\rm data}, \log\rho)$.

Figure~\ref{OddsRatioCBs} compares Sterne's and Clopper-Pearson's 95\%-confidence intervals for $\log(\rho)$ in case of a two-by-two table
\begin{equation}
	\begin{array}{|c|c||c|} \hline
		x       & 49 - x  & 49 \\\hline
		245 - x & 72 + x  & 317 \\\hline\hline
		245     & 121     & 366 \\\hline
	\end{array}
	\label{2-2-Table}
\end{equation}
for various values of $x$.

\begin{figure}
\centering
\includegraphics[width=0.7\textwidth]{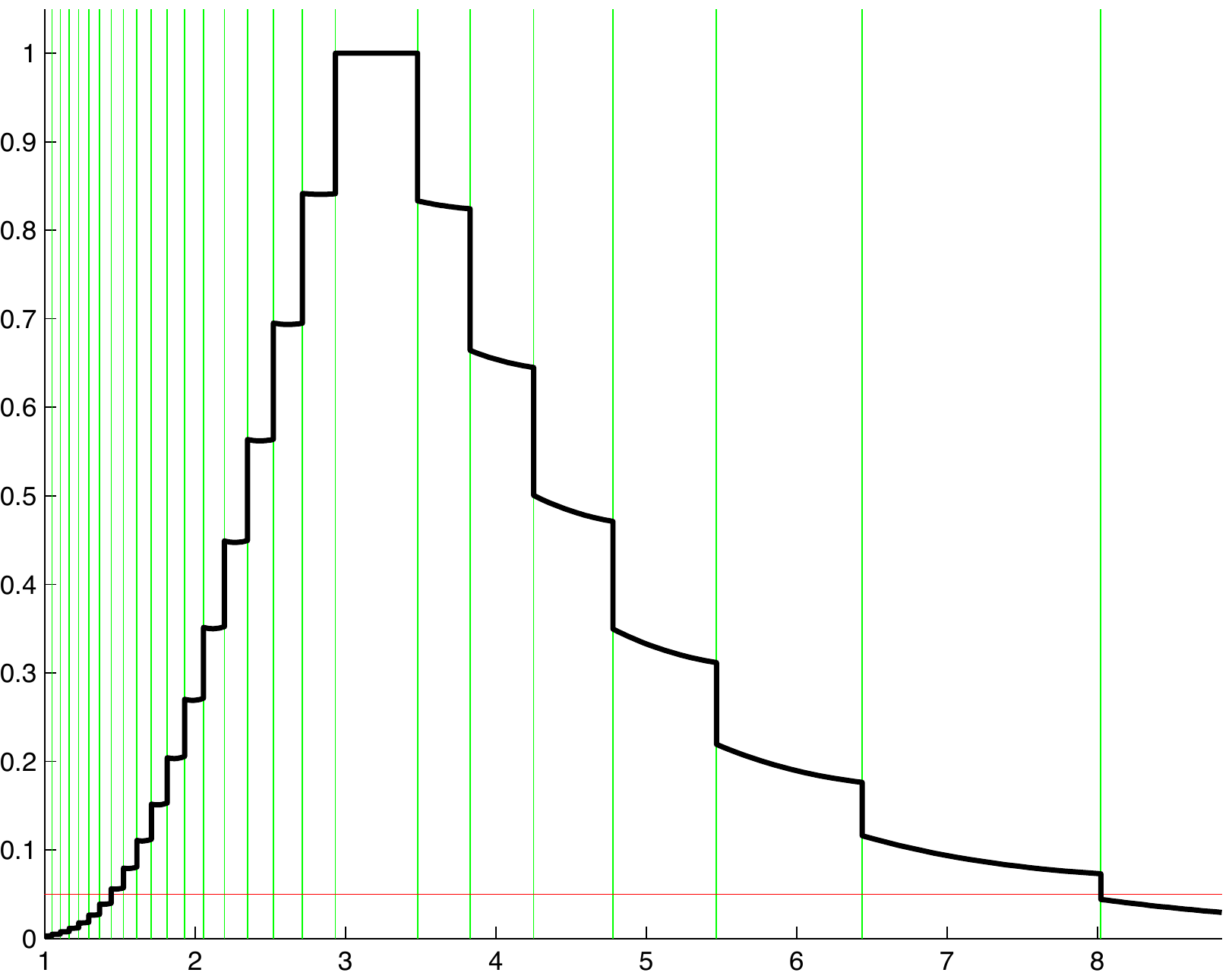}
\caption{The p-value function $\rho \mapsto \pi({\rm data},\log\rho)$ for the cervical cancer data.}
\label{OddsRatio}
\end{figure}

\begin{figure}
\centering
\includegraphics[width=0.7\textwidth]{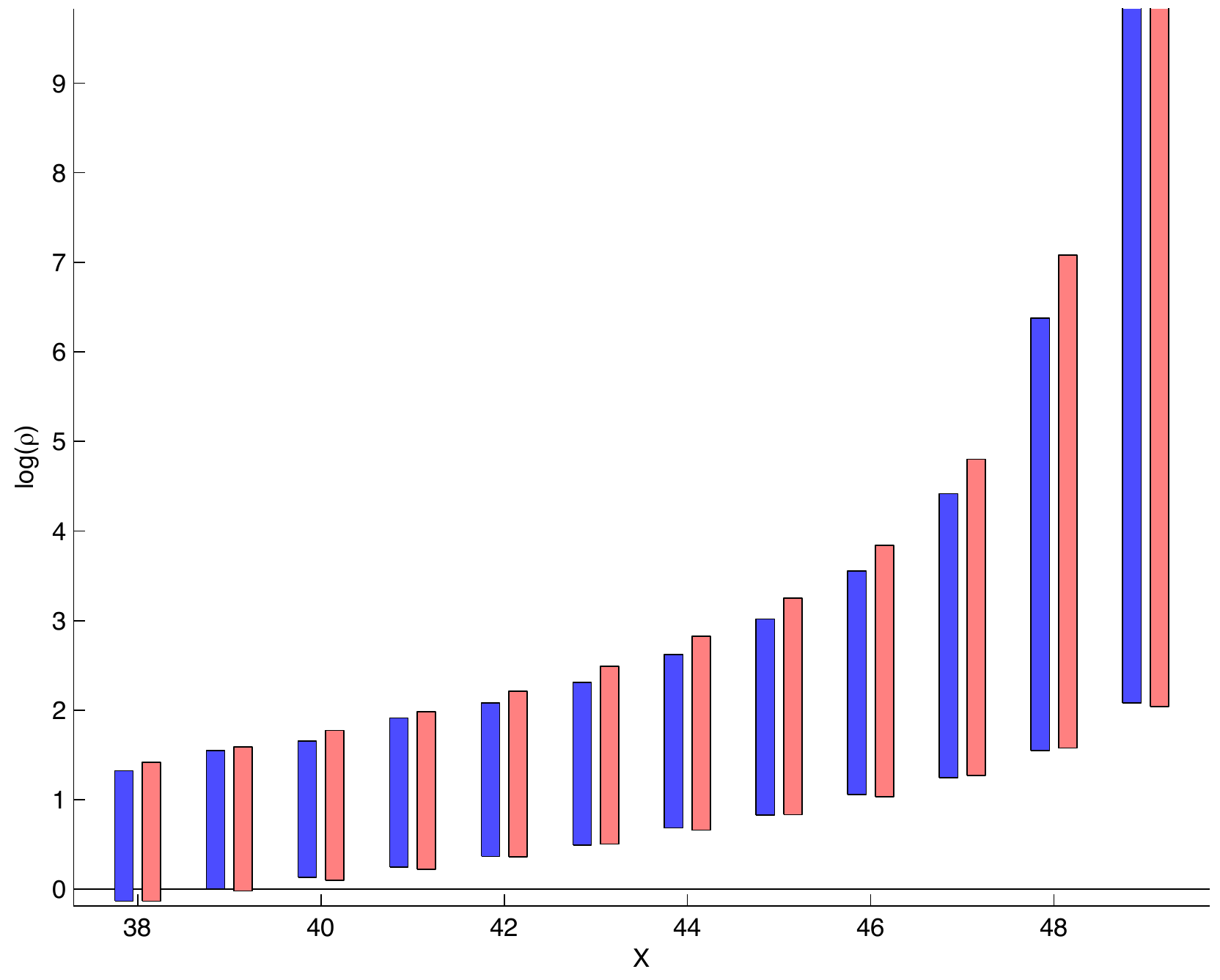}
\caption{95\%-confidence intervals for $\log(\rho)$ based on data as in \eqref{2-2-Table}.}
\label{OddsRatioCBs}
\end{figure}

\section{Proofs}
\label{Proofs}

\begin{proof}[\bf Proof of Theorem~\ref{Analysis}]
Note that for any $y \in \XX \setminus \{x\}$ and $\eta \in \mathbb{R}$, the inequality $f_\eta(y) > f_\eta(x)$ is equivalent to $\log w_y - \log w_x + \eta (y - x) > 0$, and this is satisfied if, and only if,
\[
	\begin{cases}
		\text{either} \ & y > x \ \text{and} \ \eta > \theta_{y,x} , \\
		\text{or} \     & y < x \ \text{and} \ \eta < \theta_{y,x} .
	\end{cases}
\]
Since $\theta_{k,x}$ is strictly increasing in $k \in \XX \setminus \{x\}$, we can conclude the following:

\begin{description}
\item[(i)] For $\theta_{x-1,x} \le \eta \le \theta_{x+1,x}$,
\[
	\bigl\{ y \in \XX : f_\eta(y) > f_\eta(X) \bigr\} \ = \ \emptyset .
\]
\item[(ii)] For $k < x-1$ and $\theta_{k,x} \le \eta < \theta_{k+1,x}$,
\[
	\bigl\{ y \in \XX : f_\eta(y) > f_\eta(x) \bigr\} \ = \ \{k+1,\ldots,x-1\} ,
\]
while for $k > x+1$ and $\theta_{k-1,x} < \eta \le \theta_{k,x}$,
\[
	\bigl\{ y \in \XX : f_\eta(y) > f_\eta(x) \bigr\} \ = \ \{x+1,\ldots,k-1\} .
\]
\end{description}

\noindent
These conclusions yield part~(a) of Theorem~\ref{Analysis}.

As to part~(b), we only prove that $\pi(x,\theta_{k,x})$ is strictly decreasing in $k \in \XX$, $k > x$. For the other assertion can be verified analogously or follows from symmetry considerations. Since $\pi(x,\theta_{x+1,x}) = 1 > 1 - f_{\theta_{x+2,x}}(x) = \pi(x,\theta_{x+2,x})$ in case of $x+2 \in \XX$, it suffices to show that $\pi(x,\theta_{k,x}) > \pi(x,\theta_{k+1,x})$ whenever $k > x+1$ and $k+1 \in \XX$.

For this purpose we assume without loss of generality that $x = 0$ and $w_0 = w_k = 1$. For otherwise one might consider the shifted support $\tilde{\XX} := \{y - x : y \in \XX\}$ and the transformed weights $\tilde{w}_z := \exp(\log w_{x+z} - \log w_x + z \theta_{k,x})$.

Now we have to prove that $\pi(0,0) > \pi(0,\theta_{k+1,0})$, which is equivalent to
\[
	\frac{\sum_{y=1}^{k-1} w_y}{\sum_{z \in \mathbb{Z}} w_z} \
	< \ \frac{\sum_{y=1}^{k} w_y e^{\delta y}}
		{\sum_{z \in \mathbb{Z}} w_z e^{\delta z}} ,
\]
where $\delta := \theta_{k+1,0} = - \log w_{k+1}/(k+1) > 0$. The latter assertion is equivalent to
\[
	\frac{\sum_{y=1}^{k-1} w_y}{\sum_{z \in \mathbb{Z} \setminus [1,k-1]} w_z} \
	< \ \frac{\sum_{y=1}^{k} w_y e^{\delta y}}
		{\sum_{z \in \mathbb{Z} \setminus [1,k]} w_z e^{\delta z}} .
\]
Since $\sum_{z \le 0} w_z \ge \sum_{z \le 0} w_z e^{\delta z}$, it suffices to show that
\[
	\frac{\sum_{y=1}^{k-1} w_y}{\sum_{z=k}^\infty w_z} \
	< \ \frac{\sum_{y=1}^{k} w_y e^{\delta y}}{\sum_{z=k+1}^\infty w_z e^{\delta z}} .
\]
Rearranging the enumerators and denominators yields the equivalent assertion that
\[
	A \ := \ \frac{\sum_{z = k+1}^\infty w_z e^{\delta z}}
		{\sum_{z = k}^\infty w_z}
	\ < \ B \ := \ \frac{\sum_{y=1}^{k} w_y e^{\delta y}}{\sum_{y=1}^{k-1} w_y} .
\]

When considering these fractions $A$ and $B$ more closely, one realizes that in order to maximize $A$, one should shift the probability distribution $\sum_{z=k}^\infty w_z \mathrm{Dirac}_z \big/ \sum_{z=k}^\infty w_z$ as far to the right as possible, while $B$ is minimized by shifting the probability distribution $\sum_{y=1}^k w_y \mathrm{Dirac}_y \big/ \sum_{y=1}^k w_y$ as far to the left as possible. But these shiftings should happen under the constraints that $w_k = 1$ and $\log w_{k+1} = - (k+1)\delta$ remain fixed, and $z \mapsto \log w_z$ remains concave in $z \ge 1$. Under these constraints, the following weights $\bar{w}_z$ represent an extremal configuration:
\[
	\bar{w}_z \ := \ \exp \bigl( (k - z) (k + 1) \delta \bigr) , \quad z \ge 1 .
\]
Figure~\ref{ExtremalW} scetches an example of the functions $\mathbb{N}_0 \ni z \mapsto \log w_z$ and $\mathbb{N} \ni z \mapsto \log\bar{w}_z$ when $k = 6$.

\begin{figure}
\centering
\includegraphics[width=0.7\textwidth]{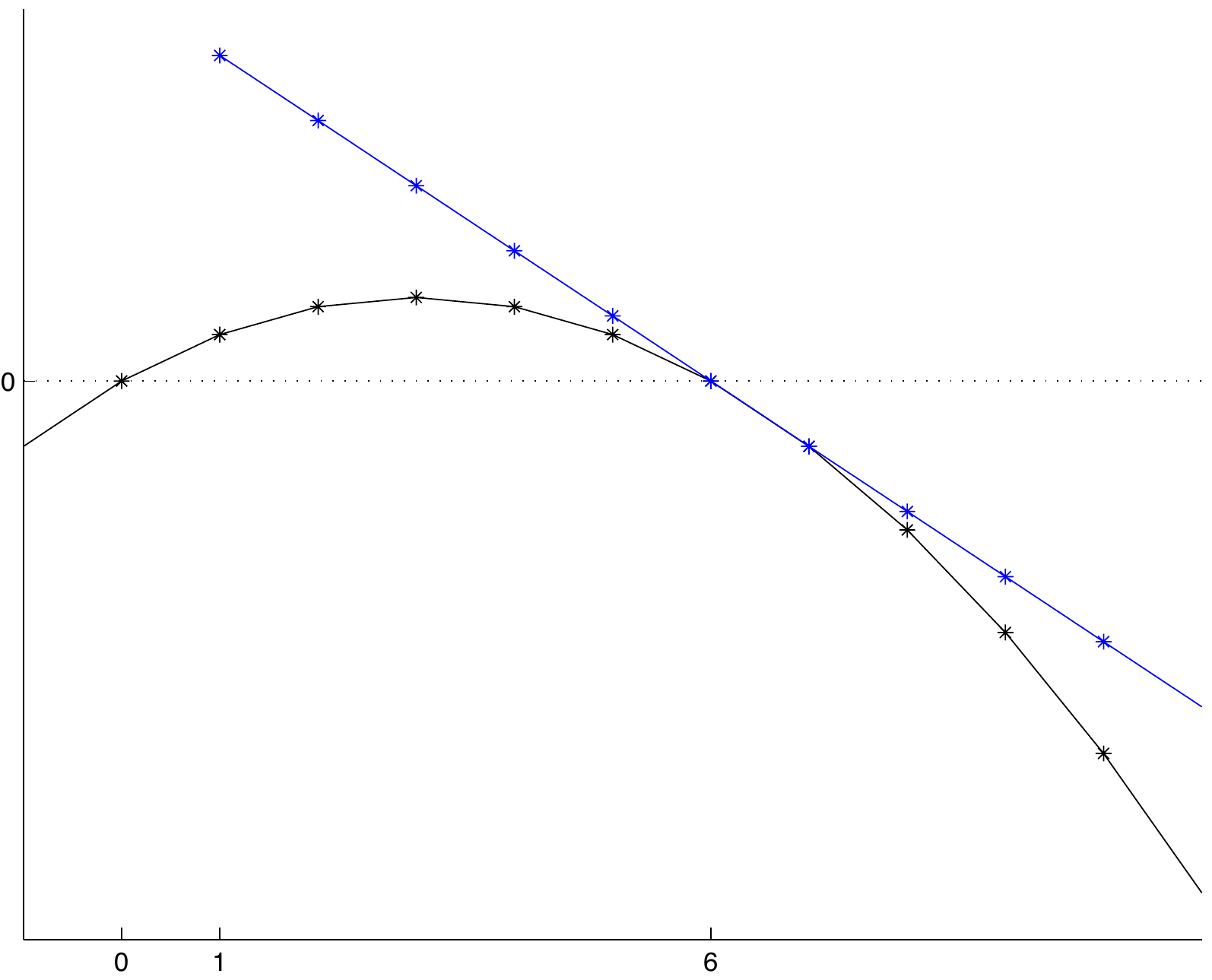}
\caption{The functions $\mathbb{N}_0 \ni z \mapsto \log w_z$ and $\mathbb{N} \ni z \mapsto \log\bar{w}_z$ in case of $k = 6$.}
\label{ExtremalW}
\end{figure}

A formal proof that replacing $w_z$ with $\bar{w}_z$ would increase $A$ and decrease $B$ uses a standard argument for distributions with monotone likelihood ratios: Note that $f_y := w_y / \sum_{z \ge k} w_z$ may be written as $r_y \bar{f}_y$ with $\bar{f}_y := \bar{w}_y / \sum_{z \ge k} \bar{w}_z$ and $r_y$ non-increasing in $y \ge k$. Moreover, $\sum_{y \ge k} f_y = \sum_{y \ge k} \bar{f}_y = 1$, while $g_y := 1_{[y > k]} e^{\delta y}$ is non-decreasing in $y \ge k$. Thus we may pick an index $a \ge k$ with $r_a \ge 1 \ge r_{a+1}$ and conclude that
\begin{align*}
	\frac{\sum_{z = k+1}^\infty w_z e^{\delta z}}{\sum_{z=k}^\infty w_z} \
	&= \ \sum_{y \ge k} f_y g_y \\
	&= \ \sum_{y \ge k} \bar{f}_y g_y + \sum_{y \ge k} (f_y - \bar{f}_y) g_y \\
	&= \ \sum_{y \ge k} \bar{f}_y g_y + \sum_{y \ge k} \bar{f}_y (r_y - 1) (g_y - g_a) \\
	&\le \ \sum_{y \ge k} \bar{f}_y g_y \ =: \ \bar{A} ,
\end{align*}
because $(r_y - 1)(g_y - g_a) \le 0$ for all $y \ge k$. But
\[
	\bar{A} \
	= \ \frac{\sum_{y \ge k+1} \exp(- y k \delta)}{\sum_{y \ge k} \exp(- y(k+1)\delta)}
	\ = \ \frac{1 - \exp(- (k+1) \delta)}{1 - \exp(- k \delta)} .
\]

With similar arguments applied to $f_y := w_y e^{\delta y} / \sum_{z = 1}^k w_z e^{\delta z}$, $\bar{f}_y := \bar{w}_y e^{\delta y} / \sum_{z = 1}^k \bar{w}_z e^{\delta z}$ and $g_y := 1\{y < k\} e^{-\delta y}$ one can show that $B$ is not smaller than
\[
	\bar{B}
	\ := \ 1 \Big/ \sum_{y=1}^k \bar{f}_y g_y
	\ = \ \frac{\sum_{y=1}^k \exp(- y k \delta)}{\sum_{y=1}^{k-1} \exp(- y (k+1) \delta)}
	\ = \ \frac{\exp(\delta) - \exp(- (k^2 - 1) \delta)}{1 - \exp(- (k^2 - 1) \delta)}
		\cdot \bar{A} ,
\]
which is strictly larger than $\bar{A}$.

It remains to prove part~(c). Note that this assertion is a consequence of the following statement: For arbitrary indices $a, b \in \XX$ with $a \le b$, the function
\[
	\theta \ \mapsto \ \log \sum_{x=a}^b f_\theta(x)
\]
is strictly concave, unless $\XX = \{a,\ldots,b\}$. A sufficient condition for strict concavity is the second derivative being strictly negative, and elementary calculations reveal the equivalent assertion that
\[
	\Var_\theta(X \,|\, X \in [a,b]) \ < \ \Var_\theta(X) ,
\]
unless $\XX = \{a,\ldots,b\}$. Since $\Var_\theta(X \,|\, X \in [a,b]) \to \Var_\theta(X)$ as $a \to \sup(\XX)$ and $b \to \inf(\XX)$, it suffices to show that
\begin{equation}
	\Var_\theta(X \,|\, X \in [a,b]) \ < \ \begin{cases}
		\Var_\theta(X \,|\, X \in [a-1,b])	&	\text{if} \ a-1 \in \XX , \\
		\Var_\theta(X \,|\, X \in [a,b+1])	&	\text{if} \ b+1 \in \XX .
	\end{cases}
	\label{Crucial}
\end{equation}
At this point it should be stressed that these inequalities, though looking plausible at first sight, may be violated in case of arbitrary weights $w_x > 0$, unless $a = b$. The concavity of $x \mapsto \log w_x$ turns out to be important. 

Now we restrict our attention to the nontrivial case that $a < b$. Let $c$ be the additional point in \eqref{Crucial}, i.e.\ $c = a-1$ or $c = b+1$. Further let $P_o := \LL_\theta(X \,|\, X \in [a,b])$ and $P = \LL_\theta(X \,|\, X \in [a,b] \cup \{c\})$. With $\lambda := P\{c\} \in (0,1)$ one may write
\[
	P \ = \ (1 - \lambda) P_o + \lambda \delta_c ,
\]
where $\delta_c$ denotes Dirac measure at $c$, and straightforward calculations yield
\[
	\Var(P) \ = \ (1 - \lambda) \Var(P_o) + \lambda(1 - \lambda) (c - {\rm Mean}(P_o))^2 .
\]
Thus it suffices to show that
\begin{equation}
	(1 - \lambda) (c - {\rm Mean}(P_o))^2 \ > \ \Var(P_o) .
	\label{Crucial2}
\end{equation}
It follows from Lemma~\ref{MomentInequality1} below that there exists a distribution $P_*$ on $\{a,\ldots,b\}$ having weights $P_*\{x\} = \exp(\gamma + \delta x)$ such that ${\rm Mean}(P_*) = {\rm Mean}(P_o)$, $\Var(P_*) \ge \Var(P_o)$ and $P_*\{x\} \ge P_o\{x\}$ for $x = a,b$. Hence it suffices to verify \eqref{Crucial2} with $P_*$ in place of $P_o$. Replacing $P$ with $(1 - \lambda)P_* + \lambda \delta_c$ does not alter the concavity of $x \mapsto \log P\{x\}$, so that \eqref{Crucial2} gets most difficult to verify in case of $P$ having weights $P\{x\} = \exp(\tilde{\gamma} + \delta x)$ for $x \in \{a,\ldots,b\} \cup \{c\}$. But then assertion \eqref{Crucial2} follows from Lemma~\ref{MomentInequality2}.
\end{proof}

\begin{Lemma}
\label{MomentInequality1}
Let $P_o$ be a probability distribution on $\{a,\ldots,b\}$ with $a < b$ such that $x \mapsto \log P_o\{x\}$ is concave and real-valued. Then there exist unique real constants $\gamma, \delta$ such that $P_*\{x\} := \exp(\gamma + \delta x)$ defines another probability distribution on $\{a,\ldots,b\}$ with ${\rm Mean}(P_*) = {\rm Mean}(P_o)$. This distribution satisfies the inequalities $\Var(P_*) \ge \Var(P_o)$ and $P_*\{x\} \ge P_o\{x\}$ for $x = a, b$.
\end{Lemma}

\begin{Lemma}
\label{MomentInequality2}
For a fixed number $\delta \in \mathbb{R}$ and any integer $m \ge 0$, let $Q_m$ be the probability distribution on $\{0,\ldots,m\}$ with weights $Q_m\{x\} = c_m e^{\delta x}$. Then $\Var(Q_m)$ is strictly increasing in $m$.
\end{Lemma}

\begin{proof}[\bf Proof of Lemma~\ref{MomentInequality1}]
For $\delta \in \mathbb{R}$ let $P_\delta$ be the probability distribution on $\{a,\ldots,b\}$ with weights $P_\delta\{x\} = \exp(\gamma_\delta + \delta x)$. It follows from elementary calculations or standard theory for exponential families that ${\rm Mean}(P_\delta)$ is continuous and strictly increasing in $\delta$. Moreover, $\lim_{\delta \to -\infty} {\rm Mean}(P_\delta) = a$ and $\lim_{\delta \to \infty} {\rm Mean}(P_\delta) = b$. Thus there exists a unique $\delta \in \mathbb{R}$ such that the means of $P_* = P_\delta$ and $P_o$ coincide.

Suppose that $P_*\{a\} < P_o\{a\}$. Since $x \mapsto \log P_o\{x\}$ is concave and $x \mapsto \log P_*\{x\}$ is linear, there exists a number $s \in (a,b)$ such that
\[
	P_*\{x\} \ \begin{cases}
		< \ P_o\{x\} & \text{if} \ x < s , \\
		> \ P_o\{x\} & \text{if} \ x > s .
	\end{cases}
\]
But this entails the contradiction that
\[
	{\rm Mean}(P_*) - {\rm Mean}(P_o)
	\ = \ \sum_{x=a}^b (x - s) \bigl( P_*\{x\} - P_o\{x\} \bigr)
	\ > \ 0 ,
\]
whence $P_*\{a\} \ge P_o\{a\}$. Analogously on can prove that $P_*\{b\} \ge P_o\{b\}$.

It remains to show that $\Var(P_*)$ is not smaller than $\Var(P_o)$. Since $P_*\{x\} \ge P_o\{x\}$ for $x = a,b$, there exist numbers $a \le r < s \le b$ such that
\[
	P_*\{x\} \ \begin{cases}
		\ge \ P_o\{x\} & \text{if} \ x \in \{a,\cdots,b\} \setminus (r,s) , \\
		\le \ P_o\{x\} & \text{if} \ x \in \{a,\cdots,b\} \cap (r,s) .
	\end{cases}
\]
Since the means of $P_o$ and $P_*$ coincide, $\Var(P_*) - \Var(P_o)$ equals
\[
	\sum_{x=a}^b x^2 \bigl( P_*\{x\} - P_o\{x\} \bigr)
	\ = \ \sum_{x=a}^b (x - r)(x - s) \bigl( P_*\{x\} - P_o\{x\} \bigr)
	\ \ge \ 0 .
\]\\[-7ex]
\end{proof}

\begin{proof}[\bf Proof of Lemma~\ref{MomentInequality2}]
In case of $\delta = 0$, $Q_m$ is the uniform distribution on $\{0,1,\ldots,m\}$ with mean $m/2$ and variance $m(m+2)/12$, which is obviously strictly increasing in $m$. Thus it suffices to consider the case $\delta \ne 0$. Here the normalizing constant $c_m$ is given by $(e^\delta - 1)/(e^{(m+1)\delta} - 1)$, and the moment generating function of $Q_m$ equals
\begin{align*}
	H_m(t) \
	&:= \ \sum_{x=0}^m e^{tx} Q_m\{x\} \\
	&= \ c_m \frac{e^{(m+1)(t+\delta)} - 1}{e^{t+\delta} - 1} \\
	&= \ c_m \frac{e^{(m+1)\delta} - 1 + e^{(m+1)\delta} (e^{(m+1)t} - 1)}
		{e^\delta - 1 + e^\delta (e^t - 1)} \\
	&= \ \frac{1 + A_m (e^{(m+1)t} - 1)}{1 + B (e^t - 1)} ,
\end{align*}
where $A_m := e^{(m+1)\delta} / (e^{(m+1)\delta} - 1)$ and $B = e^\delta / (e^\delta - 1)$. Now we compute a Taylor expansion of $H_m(t)$ at $t = 0$:
\begin{align*}
	H_m(t) \
	= & \frac{1 + A_m (m+1) t + A_m (m+1)^2 t^2/2 + O(t^3)}
		{1 + B t + B t^2/2 + O(t^3)} \\
	= & \bigl( 1 + A_m (m+1) t + A_m (m+1)^2 t^2/2 + O(t^3) \bigr) \\
	  & \cdot \ \bigl( 1 - B t - B t^2/2 + B^2 t^2 + O(t^3) \bigr) \\
	= & 1 + (A_m (m+1) - B) t
		+ \bigl( 2 B^2 - B + A_m (m+1)^2 - 2 A_m (m+1) B \bigr) t^2/2 \\
	  & + \ O(t^3)	\quad\text{as} \ t \to 0 .
\end{align*}
Thus the first and second moment of $Q_m$ are given by $A_m (m+1) - B$ and $A_m (m+1)^2 + 2 B^2 - B - 2 A_m (m+1) B$, respectively. Hence
\begin{align*}
	\Var(Q_m) \
	&= \ A_m (m+1) + 2 B^2 - B - 2 A_m (m+1) B - (A_m (m+1) - B)^2 \\
	&= \ B^2 - B + (m+1)^2 A_m (1 - A_m) \\
	&= \ B^2 - B - \frac{(m+1)^2}{2 \cosh((m+1)\delta) - 2} ,
\end{align*}
which is strictly increasing in $m$.
\end{proof}

\paragraph{Acknowledgements.}
The author is grateful to Christoph Frohn (L\"ubeck) for drawing his attention to the method of Armsen~(1955) and confidence limits for odds ratios. Many thanks also to Kaspar Rufibach (Basel) for discussions and useful hints, to Lutz Mattner (Trier) for his interest in this manuscript, and to Anna L.\ Gr{\"u}tter (Bern) for spotting some typos in previous versions.



\begin{thebibliography}{99}
\bibitem{Armsen_1955}
	\textsc{Armsen, P.} (1955).
	\newblock Tables for significance tests of $2\times 2$ contingency tables.
	\newblock \textsl{Biometrika \textbf{42}}, 494-511.
\bibitem{Baptista_Pike_1977}
	\textsc{Baptista, J.} and \textsc{M.C.\ Pike} (1977).
		Exact two-sided confidence limits for the odds ratio in a $2\times 2$ table.
		\textsl{J.\ Roy.\ Statist.\ Soc.\ (C) \textbf{26}}, 214-220.
\bibitem{Blyth_Still_1983}
	\textsc{Blyth, C.R.} and \textsc{H.A.\ Still} (1983).
	\newblock Binomial confidence intervals.
	\newblock \textsl{J.\ Amer.\ Statist.\ Assoc. \textbf{78}}, 108-116.
\bibitem{Brown_etal_2001}
	\textsc{Brown, L.}, \textsc{T.T.\ Cai} and \textsc{A.\ DasGupta} (2001).
	\newblock Interval estimation for a binomial proportion.
	\newblock \textsl{Statist.\ Science \textbf{16}}, 101-133.
\bibitem{Casella_1986}
	\textsc{Casella, G.} (1986).
	\newblock Refining binomial confidence intervals.
	\newblock \textsl{Canad.\ J.\ Statist. \textbf{14}}, 113-129.
\bibitem{Clopper_Pearson_1934}
	\textsc{Clopper, C.J. and E.S.\ Pearson} (1934).
	\newblock The use of confidence or fiducial limits illustrated in the case of the binomial. \
	\newblock \textsl{Biometrika \textbf{26}}, 404-413.
\bibitem{Clunies-Ross_1958}
	\textsc{Clunies-Ross, C.W.} (1958).
	\newblock Interval estimation for the parameter of a binomial distribution.
	\newblock \textsl{Biometrika \textbf{45}}, 275-279.
\bibitem{Crow_1956}
	\textsc{Crow, E.L.} (1956).
	\newblock Confidence intervals for a proportion.
	\newblock \textsl{Biometrika \textbf{43}}, 423-434.
\bibitem{Graham_Shotz_1979}
	\textsc{Graham, S. and W.\ Shotz} (1979).
	\newblock Epidemiology of cancer of the cervix in Buffalo.
	\newblock \textsl{J.\ Nat.\ Cancer Inst.\ \textbf{61}}, 23-27.
\bibitem{Lehmann_1986}
	\textsc{Lehmann, E.L.} (1986). 
	\newblock \textsl{Testing Statistical Hypotheses (2nd edition).}
	\newblock Wiley, New York
\bibitem{Sterne_1954}
	\textsc{Sterne, T.E.} (1954).
	\newblock Some remarks on confidence or fiducial limits.
	\newblock \textsl{Biometrika \textbf{41}}, 275-278.
\bibitem{Wilson_1927}
	\textsc{Wilson, E.B.} (1927).
	\newblock Probable inference, the law of succession, and statistical inference.
	\newblock \textsl{J.\ Amer.\ Statist.\ Assoc. \textbf{22}}, 209-212.
\end{thebibliography}
\end{document}